\newif\ifSV
\SVfalse
\newif\iffullpage
\fullpagefalse
\ifSV
\documentclass[nospthms]{svjour3}
\else
\documentclass{article}
\fi
\usepackage{a4}
\usepackage[all]{xy}
\iffullpage \usepackage{fullpage}\fi
\usepackage{qsymbols}
\usepackage{marvosym}
\usepackage{fancybox}
\usepackage{prooftree}
\usepackage{url}
\usepackage{color}
\usepackage{amsmath,wasysym}

\newcommand{\Player}{\textsf{P}}
\newcommand{\All}{\textsf{A}} 
\newcommand{\Bel}{\textsf{B}}
\newcommand{\Al}{\mbox{\footnotesize \textsf{A}}}
\newcommand{\Be}{\mbox{\footnotesize \textsf{B}}}

\newcommand{\game}[1]{ \langle #1\rangle}
\newcommand{\strat}[1]{ \llbracket #1\rrbracket}
\newcommand{\prof}[1]{ \langle\!\langle #1\rangle\!\rangle}
\newcommand{\conv}[1]{ #1\downarrow}
\newcommand{\Conv}[1]{ #1\Downarrow}
\newcommand{\GIO}{g_{1,0}}
\newcommand{\GOI}{g_{0,1}}
\newcommand{\SIO}{{\sf S_{1}}}
\newcommand{\SOI}{{\sf S_{0}}}
\newcommand{\AcBes}{{\sf{AcBes}}}
\newcommand{\SAcBes}{{\sf{SAcBes}}}
\newcommand{\BcAes}{{\sf{BcAes}}}
\newcommand{\SBcAes}{{\sf{SBcAes}}}
\newcommand{\sioa}{{s_{1,0,a}}}
\newcommand{\soia}{{s_{0,1,a}}}
\newcommand{\siob}{{s_{1,0,b}}}
\newcommand{\soib}{{s_{0,1,b}}}
\newcommand{\full}[1]{\curlyveedownarrow \hspace*{-8pt}\raisebox{6pt}{\scriptsize \it
    #1}}
\newcommand{\fullp}{\full{p}}
\newcommand{\nat}{\ensuremath{\mathbb{N}}}
\newcommand{\real}{\ensuremath{\mathbb{R}}}
 
\newcommand{\convp}{\ensuremath{\mathop{\vdash\! \raisebox{2pt}{\(\scriptstyle p\)}\! \dashv}}}
\newcommand{\sbis}{\sim_{s}}
\newcommand{\nodA}{*++[o][F]{\Al}}
\newcommand{\nodr}{*++[o][F]{\color{red}{r}}}
\newcommand{\nodB}{*++[o][F]{\Be}}
\newcommand{\fl}[1]{\ar@/^/[#1]^r \ar@/^/[d]^{d}}
\newcommand{\flr}[1]{\ar@[blue]@2@/^/[#1]^{\color{blue} r} \ar@/^/[d]^{d}}
\newcommand{\fld}[1]{\ar@/^/[#1]^r \ar@[blue]@2@/^/[d]^{\color{blue} d}}

\newenvironment{proof}[1]{\begin{quotation}\noindent\textsf{Proof:} #1}{\(\Box\)\end{quotation}}
\newtheorem{theorem}{Theorem}
\newtheorem{proposition}[theorem]{Proposition}

\newtheorem{example}[theorem]{Example}
\newtheorem{lemma}[theorem]{Lemma}

\title{Intelligent escalation and the principle of relativity}
\ifSV
\author{Pierre Lescanne}
\institute{University of Lyon, \'Ecole normale sup\'erieure de Lyon, CNRS (LIP), \\ 46 all\'ee
d'Italie, 69364 Lyon, France}
\else
\author{Pierre Lescanne\\
University of Lyon, \'Ecole normale sup\'erieure de Lyon, CNRS (LIP), \\ 46 all\'ee
d'Italie, 69364 Lyon, France}
\fi

\begin{document}
\maketitle

\begin{abstract}
\ifSV\else  \medskip\hrule\fi

\medskip

Escalation is the fact that in a game (for instance in an auction), the agents play
forever.  The $0,1$-game is an extremely simple infinite game with intelligent agents
in which escalation arises.  At the light of research on cognitive psychology, it
shows the difference between intelligence (algorithmic mind) and rationality
(algorithmic and reflective mind) in decision processes.  It also shows that
depending on the point of view (inside or outside) the perception of the rationality
of the agent change, which we call the \emph{principle of relativity}.

\ifSV\end{abstract}
\else \medskip

\end{abstract}
\hrule\fi

\bigskip 


\rightline{To \textsf{Bernard Maris}}
\rightline{\scriptsize 23 September 1946 -- 7 January 2015}
\bigskip

\rightline{\parbox{8cm}{\begin{it} That ``rational agents'' should \emph{not} engage
      in such [escalation] behavior seems obvious.
    \end{it}}}  %
  \medskip%
  \rightline{Wolfgang Leininger~\cite{leininger89:_escal_and_cooop_in_confl_situat}}

\medskip

\rightline{\parbox{8cm}{\begin{it} I can calculate the movement of the stars, \\
but not the madness of men.
\end{it}}}%
\rightline{Newton\footnote{Actually probably apocryphal Newton's view on the outcome of the
    \emph{South Sea Bubble}.} in (1720)}

\medskip

Sequential games are the natural framework for decision processes.  In this paper we
study a decision phenomenon called \emph{escalation}.  Infinite sequential games
presented here generalize naturally sequential games with perfect information and
have been introduced by Lescanne~\cite{DBLP:journals/corr/abs-0904-3528} and Lescanne
and Perrinel~\cite{DBLP:journals/acta/LescanneP12} and formalize what is proposed in
the literature~\cite{osborne94:_cours_game_theory_full_first_name}.  Sequential games
are games in which each player plays one after the other (or possibly after herself).
Here we prove, using coinduction on a simple example, namely \emph{the 0,1 game}, that
escalation is not irrational.  More precisely, in escalation, agents exhibit only the
low part of the rational mind, namely the \emph{algorithmic mind}, assimilated to
\emph{intelligence} or \emph{fluid intelligence} to be more specific (see
Stanovich~\cite{stanovich2010intelligence}).  Since intelligent  behavior contrasts
with the somewhat obvious observation that escalation is irrational, we state a
\emph{principle of relativity} which says that \emph{``the view of the insider (the
  agent) is not the same as the view of the outsider (the observer)''}, a well known
fact in computer science when studying
distributive system~\cite{jacobs12:_introd_coalg}.
In addition, the 0,1 game has nice
properties which make it an excellent paradigm of escalation, a good domain of
application of coalgebras and coinduction and a very nice opportunity to present
coinduction, a tool designed  to prove properties of complex systems.

\section{The problem of escalation}
\label{sec:escal}

Escalation in sequential games is a classic of game theory and it is admitted that
escalation is irrational.   Consider  agents able to reason formally that is making choices which are optimal
and robust, in other words choosing equilibria. 
In finite sequential games,
a right choice is obtained by a specific equilibrium called \emph{backward induction} (see Appendix).
More precisely a consequence of Aumann's theorem~\cite{aumann95} says that an agent
takes a good decision in a finite sequential game if she makes her choice
according to backward induction. In this paper we generalize backward induction into
\emph{subgame perfect equilibria} and we explore the kind of reasoning obtained  when
built on subgame perfect equilibria (\textsf{SPE} in short).  Such an appropriate reasoning
is said to be based on \emph{an algorithmic mind}  (see Section~\ref{sec:really}).

\paragraph{What is escalation?} 
In a sequential game, escalation is the possibility that agents take adequate
decisions forever without stopping.  This phenomenon has been evidenced by
Shubik~\cite{Shubik:1971} in a game called the \emph{dollar auction}. Without being
very difficult, the analysis of the dollar auction is relatively involved, because it
requires infinitely many infinite strategy profiles indexed by
$n`:\nat$~\cite{DBLP:journals/acta/LescanneP12}.  At each step there are two and only
two equilibria and therefore two potential intelligent decisions for the agents, namely ``stop''
or ``continue''.  If both agents choose always ``continue'', an escalation occurs.  In
this paper, we propose an example which is simpler theoretically and which
offers infinitely many infinite equilibria at each step unlike the dollar auction.  Due to the form of the equilibria,
the agent has no clue on which strategy is taken by her opponent.

\paragraph{What is coinduction?}
Wikipedia gives the following definition of coinduction: \emph{``In computer science,
  coinduction is a technique for defining and proving properties of systems of
  concurrent interacting objects.''}  Actually coinduction is more than a technique,
since it is a formal approach to correct reasoning on infinite structures and to
complex systems which applies far beyond computer science, especially to economics.
First traces of coinduction can be found at the beginning of the XX$^{th}$ century
when researchers tried to set a mathematical foundation to not well-founded
sets~\cite{mirimanoff17:_les_antim_de_russel_et}.  It was revisited by Peter
Aczel~\cite{aczel88:_non_well_found_sets} creating the foundation of coinduction.
Davide Sangiorgi~\cite{DBLP:journals/toplas/Sangiorgi09} gives a historical account
of the field.  Notice that Lawrence Moss and Ignacio Viglizzo have already
apply coalgebras to the theory of normal form game~\cite{DBLP:journals/entcs/MossV04}.

\paragraph{Escalation and infinite games.}

Books and
articles~\cite{colman99:_game_theor_and_its_applic,gintis00:_game_theor_evolv,osborne04a,leininger89:_escal_and_cooop_in_confl_situat,oneill86:_inten_escal_and_dollar_auction}
cover escalation.  Following Shubik, all agree that escalation takes place and can
only take place in an infinite game, but their argument uses reasoning on finite
games.  Indeed, if we cut at a finite position the infinite game of the dollar
auction in which escalation is supposed to take place, we get a finite game, in which
the only right decision is to never start the game, because the only backward
induction equilibrium corresponds to not start playing.  Then the result is
extrapolated by the authors to infinite games by making the size of the game to grow
indefinitely.  However, it has been known for a long time at least since
Weierstra\ss{}~\cite{weierstrass72}, that the ``cut and extrapolate'' method is wrong
(see Appendix), or said otherwise, there is no continuity at infinite.  For
Weierstra\ss{} this would lead to the conclusion that the infinite sum of
differentiable functions would be differentiable whereas he has exhibited a famous
counterexample.  In the case of infinite structures like infinite games, the right
reasoning is coinduction.  With coinduction we were able to show that in the dollar
auction escalation can be the result of a formal
reasoning~\cite{DBLP:journals/acta/LescanneP12,DBLP:journals/corr/abs-1112-1185}.
Currently, since the tools used generally in economics are pre-coinduction based,
they conclude that bubbles and crises are impossible and everybody's experience has
witnessed the opposite.  Careful analysis done by quantitative economists, like for
instance
Bouchaud~\cite{bouchaud08:_econom,bouchaud03:_theor_finan_risk_deriv_pricin}, have
shown that bursts (aka volatility), which share much similarities with escalation,
actually take place at any time scale.  Escalation is therefore an intrinsic feature
of economics.  Consequently, coinduction is probably the tool that economists who
call for a rethinking economics
\cite{colander05,bouchaud08:_econom,turner12:_econom_crisis} are waiting
for~\cite{winschel12:_privat}.

\paragraph{Structure of the paper}

This paper is structured as follows. In Section~\ref{sec:binary-sequ-games} we
present infinite games, infinite strategy profiles and infinite strategies, then we
describe the 0,1-game in Section~\ref{sec:0-1}. Last, we introduce the concept of
equilibrium (Sections~\ref{sec:subg-pertf-equil} and~\ref{sec:Nash}) and we discuss
escalation (Section~\ref{sec:escalation}). In Section~\ref{sec:really} we discuss the
actual rationality of the agents from a cognitive science point of view.  In an
appendix, we talk about finite $0,1$-games and finite strategy profiles. Bart
Jacobs~\cite{jacobs12:_introd_coalg}, Jan Rutten~\cite{DBLP:journals/tcs/Rutten00}
and Davide Sangiorgi~\cite{sangiorgi11} propose didactic introductions to coalgebras
and coinduction.  This paper is a deeply modified version
of~\cite{DBLP:conf/calco/Lescanne13}.

\section{Two choice sequential games}
\label{sec:binary-sequ-games}

Our aim is not to present a general theory of coalgebras or a theoretical
foundation of infinite extensive games. For this the reader is invited to look
at~\cite{Abramsky:arXiv1210.4537,DBLP:journals/corr/abs-1112-1185,DBLP:journals/acta/LescanneP12}.
But we want to give a taste of infinite sequential games\footnote{If the reader feels
  that this approach is not formal enough, she (he) can look at the ``ultra''-formal
approach found in the COQ scripts mentioned in Section~\ref{sec:Nash}.} through a very simple
one. This game has two agents and two choices.  To support our claim about the
intelligence of escalation, we do not need more features.

Assume that the set $\textsf{P}$ of agents is made of two agents called \textsf{A}
and \textsf{B}.  In this framework, an `infinite sequential two choice game' has two
shapes.  First, it can be an ending position in which case it boils down to the
distribution of the payoffs to the agents. In other words, an ending game is reduced
to a function $f: A"|->" f_{\Al}, B"|->" f_{\Be}$ and we write it $\game{f}$.  Second
, it can be a generic game with a set \textsf{Choice} made of two potential
choices: $d$ or $r$ ($d$ for \emph{down} and $r$ for \emph{right)}.  Since the game
is potentially infinite, it may continue forever.  Thus formally in this most general
configuration a game can be seen as a triple:
\begin{displaymath}
g = \game{p, g_d, g_r}.
\end{displaymath}
where $p$ is an agent and $g_d$ and $g_r$ are themselves games. The subgame $g_d$
is for the down choice, i.e., the choice corresponding to go \emph{down} and
the subgame $g_r$ is for the \emph{right} choice, i.e., the choice
corresponding to go to the right. Therefore, we define a functor
(see~\cite{DBLP:journals/tcs/Rutten00} page~4 and following):
\[\game{~}: \textsf{X} \quad "->" \quad \real^\Player \quad + \quad
\Player \times \textsf{X} \times \textsf{X}.\]
introducing a category of coalgebras
of which \textsf{Game} is the final coalgebra and where $\Player =
\{\mathsf{A}, \mathsf{B}\}$.   In other words, \textsf{Game} satisfies the isomorphism
\begin{displaymath}
  \textsf{Game} \quad \simeq \quad \real^\Player \quad + \quad
\Player \times \textsf{Game} \times \textsf{Game}.
\end{displaymath}
\begin{example}\label{ex:fingam}
Here is a picture of a typical finite sequential game:

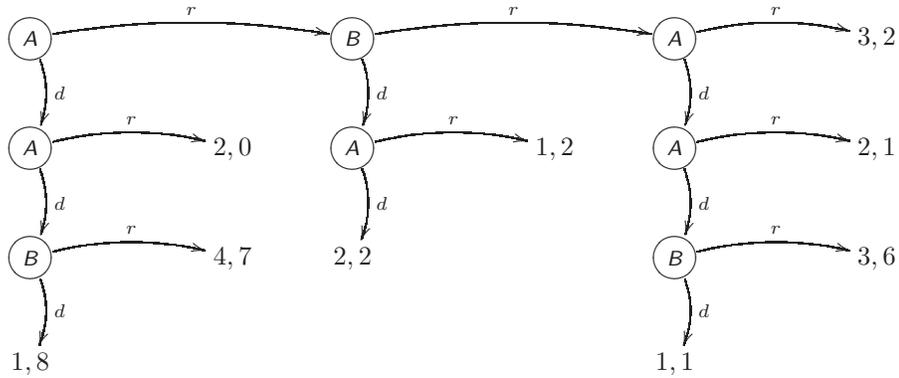
\begin{displaymath}
  \xymatrix@C=25pt{
    \nodA\fl{rrr} &&& \nodB\fl{rrr} &&& \nodA\fl{rr} && 3,2\\
    \nodA\fl{rr}&&2,0&\nodA\fl{rr}&&1,2&\nodA\fl{rr} && 2,1&&\\
    \nodB\fl{rr}&&4,7&2,2&&&\nodB\fl{rr}&&3,6\\
    1,8&&&&&&1,1&&& }
\end{displaymath}
Let us call it $gg$. In this picture $3,2$ represents the game $\game{\Al"|->"3, \Be
  "|->"2}$ and $gg=\game{\Al,gg_1, gg_2}$ where
  \begin{eqnarray*}
gg_1 &=&
  
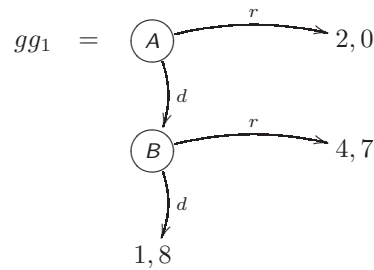
\xymatrix@C=25pt{
    \nodA\fl{rr}&&2,0\\
    \nodB\fl{rr}&&4,7\\
    1,8}
\end{eqnarray*}
and 
\begin{eqnarray*}
  gg_2 &=&
  
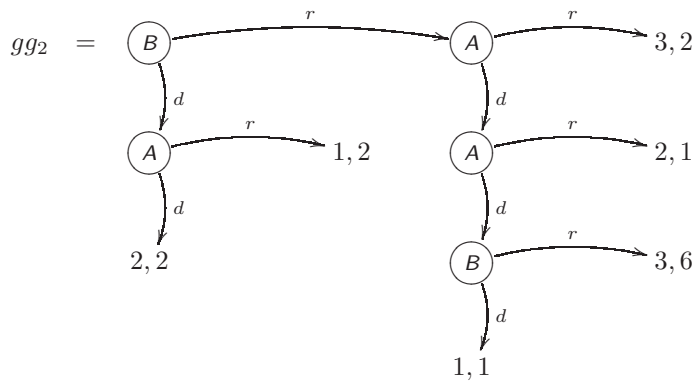
\xymatrix@C=25pt{
    \nodB\fl{rrr} &&& \nodA\fl{rr} && 3,2\\
    \nodA\fl{rr}&&1,2&\nodA\fl{rr} && 2,1&&\\
    2,2&&&\nodB\fl{rr}&&3,6\\
    &&&1,1&&& }
\end{eqnarray*}
If we write the full decomposition of $gg_1$ we get:
\begin{displaymath}
  gg_1 = \game{\All, \game{\Bel, \game{\Al"|->"1, \Bel
        "|->"8}, \game{\Al"|->"4, \Bel "|->"7}}, \game{\All"|->"2, \Bel "|->"0}}.
\end{displaymath}
\end{example}
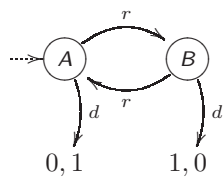
\begin{example}
  \begin{displaymath}
      \xymatrix@C=10pt{
          &\ar@{.>}[r]& *++[o][F]{\Al} \ar@/^1pc/[rr]^r \ar@/^/[d]^d 
          &&*++[o][F]{\Be} \ar@/^1pc/[ll]^r \ar@/^/[d]^d \\
          &&0,1&&1,0
        }
  \end{displaymath}
is a picture of an infinite game which will be studied more formally in
Section~\ref{sec:0-1}.  Notice the dotted arrow 
$\xymatrix@C=10pt{\ar@{.>}[r]&}$ which
shows the start of the game.
\end{example}
\begin{example}\label{ex:flower}
  The game which is solution of the equation
  \begin{eqnarray*}
    g &=& \game{\All, \game{\Bel, g, g},\game{\Bel, g, g}}
  \end{eqnarray*}
  is a game infinite in both direction, \emph{down} and \emph{right}, alternating
  agents \Al{} and \Be{}.  Notice that it has no leaf and that payoffs are not
  attributed.  It can be pictured as:
  
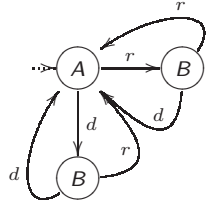
\begin{displaymath}
\xymatrix@C=8pt{\ar@{.>}[r]&*++[o][F]{\Al}\ar[rr]^r \ar[d]^d&& *++[o][F]{\Be} \ar
  @(ur,ur)[ll]_r \ar@(d,dr)[ll]_d&\\
  & *++[o][F]{\Be} \ar@(dl,dl)[u]^d \ar@(r,dr)[u]^r
}
\end{displaymath}

\end{example}
\begin{example}
\newcommand{\tcent}{}
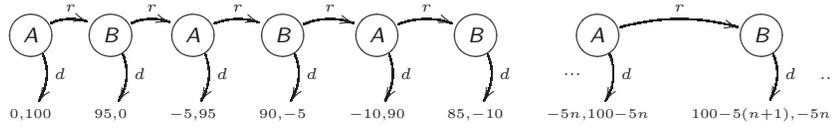
\begin{footnotesize}
  \begin{displaymath}
    \xymatrix@C=10pt@R=15pt{
      *++[o][F]{\Al} \ar@/^/[r]^r \ar@/^/[d]^d &*++[o][F]{\Be} \ar@/^/[r]^r \ar@/^/[d]^d
      &*++[o][F]{\Al} \ar@/^/[r]^r \ar@/^/[d]^d &*++[o][F]{\Be} \ar@/^/[r]^r \ar@/^/[d]^d 
      &*++[o][F]{\Al} \ar@/^/[r]^r \ar@/^/[d]^d &*++[o][F]{\Be} \ar@/^/[d]^{d}
      &*++[o][F]{\Al} \ar@/^/[r]^r\ar@/^/[d]^d_{\ldots\ \ \ }&*++[o][F]{\Be}\ar@/^/[d]^{d~~~\ldots}\\
      \scriptscriptstyle{0,100}&\scriptscriptstyle{95\tcent,0}&\scriptscriptstyle{-5\tcent,95\tcent}&\scriptscriptstyle{90\tcent,-5\tcent}&\scriptscriptstyle{-10\tcent,90\tcent}&\scriptscriptstyle{85\tcent,-10\tcent}&\scriptscriptstyle{-5n\tcent,100-5n\tcent}
      &\scriptscriptstyle{100-5(n+1)\tcent,-5n\tcent}
    }
  \end{displaymath}
\end{footnotesize}
is the \emph{dollar auction} game~\cite{Shubik:1971} with bids of $5\cent$.
\end{example}

\subsection{Strategy profiles}
\label{sec:prof}

From a game, we can deduce \emph{strategy profiles} (later we will also say
sometimes simply \emph{profiles}) obtained by adding a label, at each node, which is
a choice made by the agent.  In a two choice sequential game, choices belong to the
set $\{d, r\}$.  Therefore a strategy profile which does
not correspond to an ending game is a quadruple:
\[s = \prof{p,c,s_d,s_r},\] where $p$ is an agent ($\All$ or $\Bel$), $c$ is a choice
($d$ or $r$), and, $s_d$ and $s_r$ are two strategy profiles.  The strategy profile
which corresponds to an ending position has no choice, namely it is reduced to a
function $\prof{f} = \prof{{A"|->" f_{\Al}}, {B"|->" f_{\Be}}}$.   The functor 
\begin{displaymath}
  \prof{~} : \mathsf{X} \quad "->"  \quad \real^\Player \quad + \quad 
\Player \times \textsf{Choice} \times \textsf{X} \times \textsf{X}.
\end{displaymath}
where
\begin{eqnarray*}
  \Player &=& \{\All, \Bel\}\\
  \textsf{Choice} &=& \{r, d\}
\end{eqnarray*}
introduces a category of coalgebras in which the coalgebra \textsf{StratProf} of
infinite strategy profiles is the final coalgebra.  Hence  \textsf{StratProf}
satisfies the isomorphism:
\begin{displaymath}
   \textsf{StratProf} \quad \simeq  \quad \real^\Player \quad + \quad 
\Player \times \textsf{Choice} \times \textsf{StratProf} \times \textsf{StratProf}.
\end{displaymath}
\begin{example} \label{ex:stratprof}
Here are the pictures of three strategy profiles associated with the game
  of Figure~\ref{ex:fingam}.
  \begin{displaymath}
s_1 \quad =\quad
    
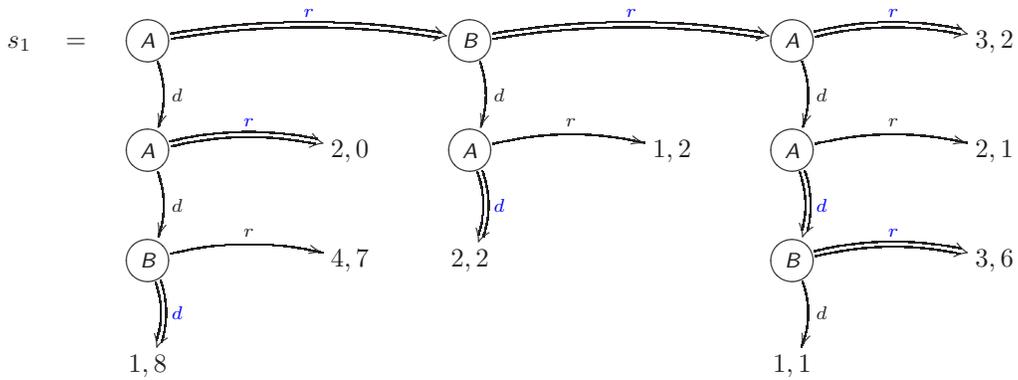
\xymatrix@C=25pt{
\nodA\flr{rrr} &&& \nodB\flr{rrr} &&& \nodA\flr{rr} && 3,2\\
\nodA\flr{rr}&&2,0&\nodA\fld{rr}&&1,2&\nodA\fld{rr} && 2,1&&\\
\nodB\fld{rr}&&4,7&2,2&&&\nodB\flr{rr}&&3,6\\
1,8&&&&&&1,1&&&
}
\end{displaymath}

\begin{displaymath}
s_2 \quad = \quad
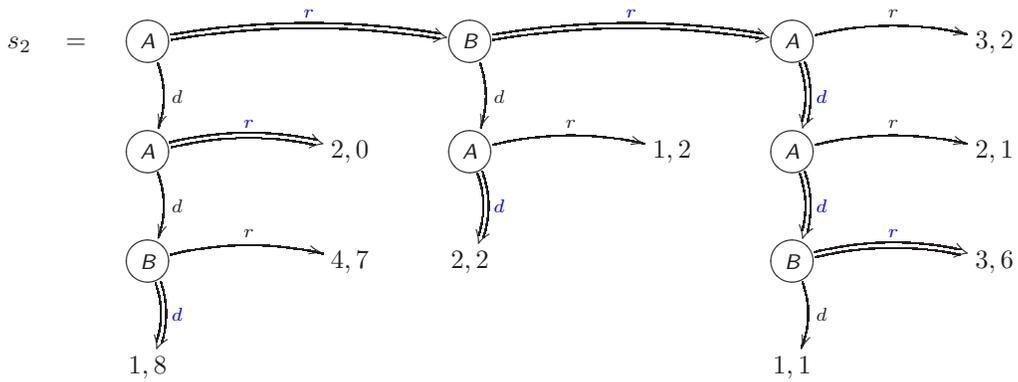
\xymatrix@C=25pt{
\nodA\flr{rrr} &&& \nodB\flr{rrr} &&& \nodA\fld{rr} && 3,2\\
\nodA\flr{rr}&&2,0&\nodA\fld{rr}&&1,2&\nodA\fld{rr} && 2,1&&\\
\nodB\fld{rr}&&4,7&2,2&&&\nodB\flr{rr}&&3,6\\
1,8&&&&&&1,1&&&
}
\end{displaymath}

\begin{displaymath}
s_3  \quad = \quad
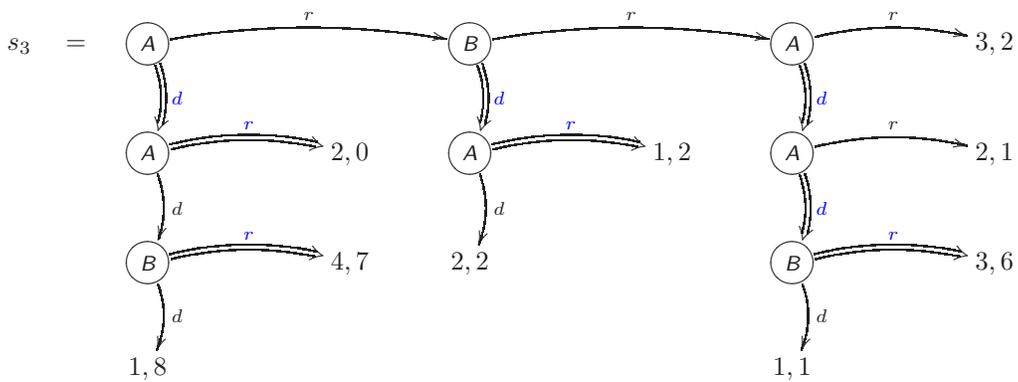
\xymatrix@C=25pt{
\nodA\fld{rrr} &&& \nodB\fld{rrr} &&& \nodA\fld{rr} && 3,2\\
\nodA\flr{rr}&&2,0&\nodA\flr{rr}&&1,2&\nodA\fld{rr} && 2,1&&\\
\nodB\flr{rr}&&4,7&2,2&&&\nodB\flr{rr}&&3,6\\
1,8&&&&&&1,1&&&
}
\end{displaymath}

  The start of those strategy profiles is taken by player $\All$ and
  she chooses \emph{right} in the two first strategy profiles and \emph{down} in the
  third strategy profile.

$s_1$ is built with the strategy profiles:
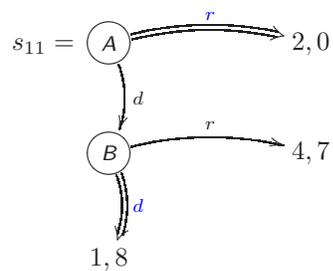
\begin{displaymath}
  s_{11} =     \xymatrix@C=25pt{
\nodA\flr{rr}&&2,0\\
\nodB\fld{rr}&&4,7\\
1,8&&
}
\end{displaymath}
and
\begin{displaymath}
  s_{12} =     
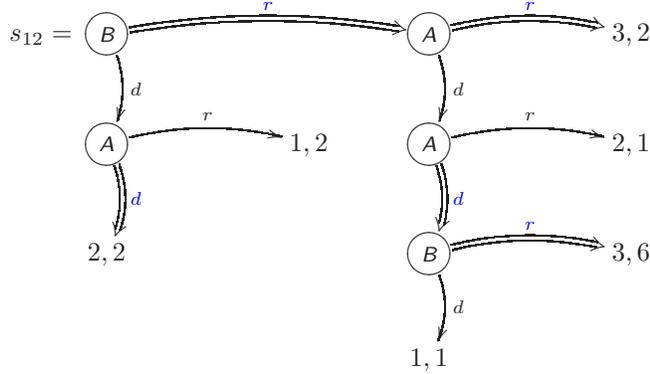
\xymatrix@C=25pt{
\nodB\flr{rrr} &&& \nodA\flr{rr} && 3,2\\
\nodA\fld{rr}&&1,2&\nodA\fld{rr} && 2,1&&\\
2,2&&&\nodB\flr{rr}&&3,6\\
&&&1,1&&&
}
\end{displaymath}
with
\begin{displaymath}
  s_1 = \prof{\All, r, s_{11}, s_{12}}.
\end{displaymath}
The full decomposition of $s_{11}$ is
\begin{displaymath}
  s_{11}= \prof{\All, r, \prof{\Bel, d, \prof{\Al"|->"1, \Bel
        "|->"8}, \prof{\Al"|->"4, \Bel "|->"7}}, \prof{\All"|->"2, \Bel "|->"0}}.
\end{displaymath}

\end{example}
\begin{example}
  Figure~\ref{fig:boucle} on page~\pageref{fig:boucle}, Figure~\ref{fig:alwr} on
  page~\pageref{fig:alwr}, and Figure~\ref{fig:ABAB} on page~\pageref{fig:ABAB} give
  strategy profiles of infinite sequential games.
\end{example}
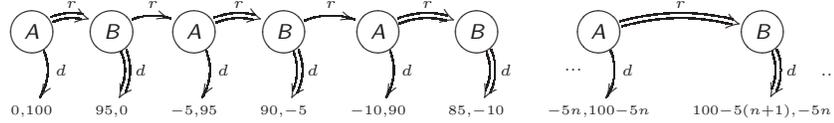
\begin{example}
\newcommand{\tcent}{}
\begin{footnotesize}
  \begin{displaymath}
    \xymatrix@C=10pt@R=15pt{
      *++[o][F]{\Al} \ar@[blue]@{=>}@/^/[r]^r \ar@/^/[d]^d &*++[o][F]{\Be} \ar@/^/[r]^r \ar@[blue]@{=>}@/^/[d]^d
      &*++[o][F]{\Al} \ar@[blue]@{=>}@/^/[r]^r \ar@/^/[d]^d &*++[o][F]{\Be} \ar@/^/[r]^r \ar@[blue]@{=>}@/^/[d]^d 
      &*++[o][F]{\Al} \ar@[blue]@{=>}@/^/[r]^r \ar@/^/[d]^d &*++[o][F]{\Be} \ar@[blue]@{=>}@/^/[d]^{d}
      &*++[o][F]{\Al} \ar@[blue]@{=>}@/^/[r]^r\ar@/^/[d]^d_{\ldots\ \ \ }&*++[o][F]{\Be}\ar@[blue]@{=>}@/^/[d]^{d~~~\ldots}\\
      \scriptscriptstyle{0,100}&\scriptscriptstyle{95\tcent,0}&\scriptscriptstyle{-5\tcent,95\tcent}&\scriptscriptstyle{90\tcent,-5\tcent}
      &\scriptscriptstyle{-10\tcent,90\tcent}&\scriptscriptstyle{85\tcent,-10\tcent}
      &\scriptscriptstyle{-5n\tcent,100-5n\tcent}
      &\scriptscriptstyle{100-5(n+1)\tcent,-5n\tcent}
    }
  \end{displaymath}
\end{footnotesize}
is a strategy profile which is of interest in the dollar auction
game~\cite{DBLP:journals/acta/LescanneP12} for proving that agent reasoning formally
can enter escalation.
\end{example}
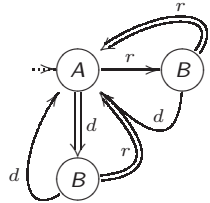
\begin{example}
  \begin{displaymath}
\xymatrix@C=8pt{\ar@{.>}[r]&*++[o][F]{\Al}\ar[rr]^r \ar@[blue]@{=>}[d]^d&& *++[o][F]{\Be} \ar@[blue]@{=>}
  @(ur,ur)[ll]_r \ar@(d,dr)[ll]_d&\\
  & *++[o][F]{\Be} \ar@(dl,dl)[u]^d \ar@[blue]@{=>}@(r,dr)[u]^r
}
\end{displaymath}
is a strategy profile of the game of Example~\ref{ex:flower} and  is solution of
the  equation
  \begin{eqnarray*}
    s &=& \prof{\All, d,\prof{\Bel, r, s, s},\prof{\Bel, r, s, s}}
  \end{eqnarray*}
\end{example}
From a strategy
profile, we
can build a game by removing the choices:
\begin{eqnarray*}
\mathsf{game} &::& \mathsf{StratProf} "->" \mathsf{Game}\\
 \mathsf{game}(\prof{f}) &=& \game{f}\\
\mathsf{game}(\prof{p,c,s_d,s_r}) &=& \game{p,\mathsf{game}(s_d), \mathsf{game}(s_r)}
\end{eqnarray*}
$\mathsf{game}(s)$ is the game of the strategy profile $s$.  Notice that the function
\textsf{game} is not recursive like say the function \textsf{fib}
\begin{eqnarray*}
  \mathsf{fib}(0) &=& 0 \\
\mathsf{fib}(1) &=& 1 \\
\mathsf{fib}(n+2) &=& \mathsf{fib}(n+1) + \mathsf{fib}(n)
\end{eqnarray*}
which defines the Fibonacci sequence. In fact, \textsf{game} is corecursive since it
works on potentially infinite structures (see the end of this section) whereas
\textsf{fib} which works on natural numbers works on finite structure and is
recursive.  Notice however that if we consider only finite games and finite strategy
profiles, then a recursive \textsf{game} function could be defined.  In this paper we
 consider the final coalgebras \textsf{Game} which contains finite and infinite
sequential games and the final coalgebra \textsf{Stratprof} which contains finite and
infinite strategy profiles.

Given a strategy
profile $s$, we can associate, by induction, a (partial) payoff function $\widehat{s}$,
as follows:
\begin{displaymath}
\begin{array}{lll}
\textit{~when~} &s = \prof{f}           &\quad \widehat{s} \quad =\quad f\\
\textit{~when~} & s = \prof{p,d,s_d,s_r} &\quad \widehat{s} \quad =  \quad \widehat{s_d} \qquad \\
\textit{~when~} & s = \prof{p,r,s_d,s_r} & \quad\widehat{s} \quad =  \quad\widehat{s_r} \qquad 
\end{array}
\end{displaymath}
In the literature on extensive games
(\cite{aumann95,halpern01:_subst_ration_backw_induc} for instance), authors have a
graph vision and they write $h^v_p(s)$ the payoff obtained by $p$ starting at vertex
$v$.  Here we consider only the start vertex (let us call it \emph{start}). For us,
the vertices are not primitive objects. The other vertices are start vertices of
subgames and of subprofiles and are only considered when dealing with those subgames
and those subprofiles.  What we write $\widehat{s}(p)$ would be written
$h^{\mathit{start}}_p(s)$ in~\cite{aumann95,halpern01:_subst_ration_backw_induc}.
\begin{example}
  In Example~\ref{ex:stratprof}, we have:
  \begin{displaymath}
    \begin{array}[h]{lcl}
    \widehat{s_1}(A)&=& 3\\
    \widehat{s_1}(B)&=& 2
  \end{array}
\qquad\qquad\qquad\qquad
  \begin{array}[h]{lcl}
    \widehat{s_2}(A)&=& 3\\
    \widehat{s_2}(B)&=& 6
  \end{array}
\qquad\qquad\qquad\qquad
  \begin{array}[h]{lcl}
    \widehat{s_3}(A)&=& 2\\
    \widehat{s_3}(B)&=& 0
  \end{array}
  \end{displaymath}

\end{example}
$\widehat{s}$ is not defined if its definition runs in an infinite process.  For
instance, in Example~\ref{ex:alwr}, $\widehat{s_{\Box r}}$ is not defined and in
Section~\ref{sec:escalation}, $\widehat{s_{\Al, \infty}}$ is not defined.  To ensure
that we consider only strategy profiles where the payoff function is defined we
can impose strategy profiles to be \emph{convergent}\footnote{Called
  \emph{leads to a leaf} in~\cite{DBLP:journals/acta/LescanneP12}.}, written
$\conv{s}$ (or prefixed $\downarrow(s)$) and defined as the least predicate
satisfying
\begin{displaymath}
  \conv{s} \hspace{15pt} "<=(i)>" \hspace{15pt}
  s = \prof{f} \quad \vee \quad 
  (s = \prof{p,d, s_d, s_r} \wedge \conv{s_d}) \ \vee \ 
  (s = \prof{p,r, s_d, s_r} \wedge \conv{s_r}).
\end{displaymath}
\begin{proposition}
  If  $\conv{s}$, then $\widehat{s}$ is defined.
\end{proposition}
\begin{proof}
  By induction.  If $s=\prof{f}$, then since $\widehat{s} =f$ and $f$ is defined,
  $\widehat{s}$ is defined.

  Assume $\conv{s}$.  If $s = \prof{p,c, s_d, s_r}$, there are two cases: $c=d$ or
  $c=r$. Let us look at $c=d$.  If $c=d$, then $\conv{s_d}$ and $\widehat{s_d}$ is
  defined by induction and since $\widehat{s}=\widehat{s_d}$, we conclude that
  $\widehat{s}$ is defined.

  The case $c=r$ is similar.
\end{proof}
As we will consider the payoff function also for subprofiles, we want the payoff
function to be defined on subprofiles as well.  Therefore we define a property
stronger than convergence which we call \emph{strong convergence}\footnote{Called
  \emph{always lead to a leaf} in~\cite{DBLP:journals/acta/LescanneP12}.}.  We say
that a strategy profile $s$ is \emph{strongly convergent} and we write it $\Conv{s}$
if it is the largest predicate fulfilling the following conditions.
\begin{itemize}
\item $\Conv{\prof{p,c, s_d, s_r}}$ if
\begin{itemize}
  \item $\prof{p,c, s_d, s_r}$ is convergent,
  \item $s_d$ is strongly convergent,
  \item $s_r$ is strongly convergent.
  \end{itemize}
\item $\Conv{\prof{f}}$ that is that for whatever $f$, $\prof{f}$ is strongly convergent
\end{itemize}
More formally:
\[\Conv{s} \quad "<=(c)>" \quad s = \prof{f} \ \vee \ (s = \prof{p,c, s_d, s_r} \wedge
\conv{s} \wedge \Conv{s_d} \wedge \Conv{s_r}).\]

There is however a difference between the definitions of $\downarrow$ and
$\Downarrow$. Wherever $\conv{s}$ is defined \emph{by induction}\footnote{Roughly
  speaking a definition by induction %
  works from the basic elements, to the constructed elements.   For the natural
  numbers, for $0$ and from $n$ to $n+1$. For finite strategy profiles, for $\prof{f}$ and
  from $s_1$ and $s_2$ to $\prof{p,c,s_1,s_2}$.}, %
from ending games to the game $s$, $\Conv{s}$ is defined \emph{by
  coinduction}\footnote{Roughly speaking a definition by coinduction works on
  infinite objects, like an invariant.}.  This difference between recursive and
corecursive definition is the core of coalgebra
theory~\cite{jacobs12:_introd_coalg,DBLP:journals/tcs/Rutten00,sangiorgi11}.  However, this is not the aim of the present paper to
present coinduction in detail. 
Both induction and coinduction are based on the fixed-point theorem.
The definition of $\Downarrow$ is typical of infinite profiles and means that $\Downarrow$ is invariant along the
infinite profile.  To make the difference clear and explicit between the definitions, we use the
symbol \("<=(i)>"\) for inductive definitions and the symbol \("<=(c)>"\) for
coinductive definitions.  Recall that the definition of the function
\mbox{\textsf{game} :: \textsf{StratProf} $"->"$ \textsf{Game}}
was presented as a coinductive function.  We get easily the following proposition.
\begin{proposition}
  $\Conv{s} \quad "=>" \quad \conv{s}.$
\end{proposition}
Clearly we do not have the opposite implication, as shown by
Example~\ref{ex:ddBoxr}.  Indeed $\conv{s}$ is a local property whereas $\Conv{s}$ is a
somewhat global property. 

We can define  the notion of \emph{subprofile},
written $\precsim$:
\[s' \precsim s \quad "<=(i)>" \quad s' \sbis s \ \vee\ s = \prof{p, c, s_d, s_r} \wedge
(s' \precsim s_d \vee s' \precsim s_r),\] %
where $\sbis$ is the bisimilarity\footnote{The reader can consider $\sbis$ as the
  equality on \textsf{StratProf}.} among profiles defined as the largest binary
predicate $s' \sbis s$ such that
\iffullpage
\[
s'\sbis s\quad"<=(c)>"\quad
s' =\prof{f}=s \quad \vee \quad (s' = \prof{p, c, s_d', s_r'}
\wedge  s = \prof{p, c, s_d, s_r} \wedge s_d'\sbis s_d\wedge s_r' \sbis s_r).
\]
\else
\begin{eqnarray*}
  s'\sbis s\quad"<=(c)>"\quad
s' =\prof{f}=s &\vee& (s' = \prof{p, c, s_d', s_r'} \wedge  \\
&&s = \prof{p, c, s_d, s_r} \wedge \\
&& s_d'\sbis s_d\wedge s_r' \sbis s_r).
\end{eqnarray*}
\fi Notice that since we work with infinite objects, we may have $s\not\sbis s'$ and
${s\precsim s' \precsim s}$.  In other words, an infinite profile can be a strict
subprofile of itself.  This is the case for $\sioa$ and $\siob$ in
Section~\ref{sec:subg-pertf-equil}.  If a profile is strongly convergent, then its
subprofiles are strongly convergent as well and the payoffs associated with all its
subprofiles are defined.
\begin{proposition}~
  \begin{enumerate}
  \item If $\Conv{s_1}$ and $s_2 \precsim s_1$ then $\Conv{s_2}$.
  \item If $\Conv{s_1}$ and if $s_2 \precsim s_1$, then $\widehat{s_2}$ is defined.
  \end{enumerate}

\end{proposition}

\subsection{The always modality}

We notice that $\downarrow$ characterizes a profile by a property of the start vertex,
we would say that this property is local. \(\Downarrow\) is obtained by distributing
the property along the game.  In other words we transform the predicate
\(\downarrow\) and such a predicate transformer is called a \emph{modality}.  Here we
are interested by the modality \emph{always}, also written \(\Box\).

Given a predicate $`F$ on strategy profiles, the predicate $\Box\,`F$ is defined
coinductively as follows:

\[\Box\,`F(s) \quad "<=(c)>" \quad `F(s) \wedge s = \prof{p, c, s_d, s_r} "=>" (\Box\,`F(s_d) \wedge
\Box `F(s_r)).\]
The predicate ``is strongly convergent'' is the same as the predicate ``is always
convergent''. 
\begin{proposition}
  \(\Conv{s} \qquad "<=>" \qquad \Box\downarrow(s).\)
\end{proposition}
\subsection{Strategies}
\label{sec:strat}

The coalgebra \textsf{Strat} of \emph{strategies}\footnote{A strategy is not the same as a strategy
  profile, which is obtained as the sum of strategies.} is defined by the
functor \[\strat{~}: \mathsf{X} "->" \real^{\textsf{P}} \ + \ (\Player + \textsf{Choice})
\times \mathsf{X} \times \mathsf{X}\] 
where $\Player = \{\All, \Bel\}$ and
$\textsf{Choice} =\{d,r\}$.  In other words, the coalgebra \textsf{Strat} of strategies is solution
of the equation:
\begin{displaymath}
   \textsf{Strat} \quad \simeq \quad \real^{\textsf{P}} \ + \ (\Player + \textsf{Choice})
\times \textsf{Strat} \times \textsf{Strat}.
\end{displaymath}
A strategy of agent~$p$ is a game in which some occurrences of
$p$ are replaced by choices.  A strategy is written $\strat{f}$ or $\strat{x,s_1,
  s_2}$.  
\begin{example}
  Consider the strategy for $\All$ in the game $gg$ in which $\All$ decides to always
  take the choice $r$.
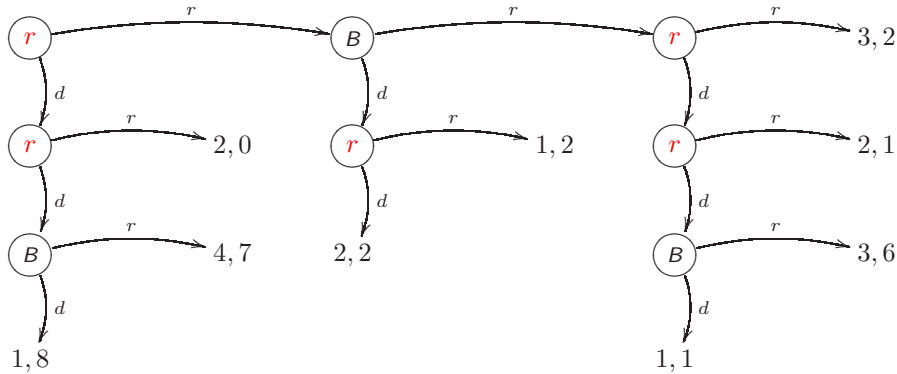
\begin{displaymath}
  \xymatrix@C=25pt{
    \nodr\fl{rrr} &&& \nodB\fl{rrr} &&& \nodr\fl{rr} && 3,2\\
    \nodr\fl{rr}&&2,0&\nodr\fl{rr}&&1,2&\nodr\fl{rr} && 2,1&&\\
    \nodB\fl{rr}&&4,7&2,2&&&\nodB\fl{rr}&&3,6\\
    1,8&&&&&&1,1&&& }
\end{displaymath}
\end{example}

By replacing the choice made by agent $p$ by the agent $p$ herself, the function \textsf{st2g}
associates a game with a pair consisting of a strategy and an agent:
\begin{displaymath}
\begin{array}{l@{\quad}c@{\quad}ll@{~}l}
  \mathsf{st2g}(\strat{f}, p) &=& \game{f}\\
  \mathsf{st2g}(\strat{x, st_1, st_2}, p) &=& \mathbf{if~} x \in \Player
  &\mathbf{~then~} &\game{x, \mathsf{st2g}(st_1, p), \mathsf{st2g}(st_2, p)} \\
  &&&\mathbf{~else~} &\game{p, \mathsf{st2g}(st_1, p), \mathsf{st2g}(st_2, p)}.
\end{array}
\end{displaymath}

If a strategy $st$ is really the strategy of agent $p$ it should contain nowhere $p$
and should contain a choice $c$ instead.  In this case we say that $st$ is \emph{full for}
$p$ and we write it $st\fullp$.
\begin{eqnarray*}
  \strat{f} \fullp \\
  \strat{x, st_1, st_2}\fullp &"<=(c)>"& (x`;\textsf{Choice}"=>" x \neq p) \wedge st_1\fullp \wedge st_2\fullp.
\end{eqnarray*}
We can sum strategies to make a profile. But for that we have to assume that
all strategies are full and have the same game. We say that the strategies are
\emph{consistent}. 
In other words, $(st_p)_{p`:\Player}$
is a family of strategies such that:
\begin{itemize}
\item $`A p`:\Player, st_p \fullp$,
\item  there exists a game $g$ such that for all $p$ in $\Player$,  \textsf{st2g}
  returns $g$, more formally  $`E g`:\mathsf{Game},`A p`:\Player, \mathsf{st2g}(st_p) = g$.
\end{itemize}
We define the sum
\begin{math} \displaystyle
  \bigoplus_{p`:\Player} st_p
\end{math}
of consistent strategies as follows:
\begin{eqnarray*}
  \bigoplus_{p`:\Player} \strat{f} &=& \prof{f}\\
  \strat{c, st_{p',1}, st_{p',2}} \oplus \bigoplus_{p`:\Player\setminus p'} \strat{p', st_{p,1}, st_{p,2}}  &=& \prof{p',c,   \bigoplus_{p`:\Player} st_{p,1},   \bigoplus_{p`:\Player} st_{p,2}}.
\end{eqnarray*}
We can show that the game  of all the strategies is the game of the
strategy profile which is the sum of the strategies.
\begin{proposition}
  $\mathsf{st2g}(st_{p'},p') = \mathsf{game}(\displaystyle \bigoplus_{p`:\Player} st_p).$
\end{proposition}

\section{Infinipede games and the 0,1-game}
\label{sec:0-1}

We will restrict to simple games which have the shape of combs,
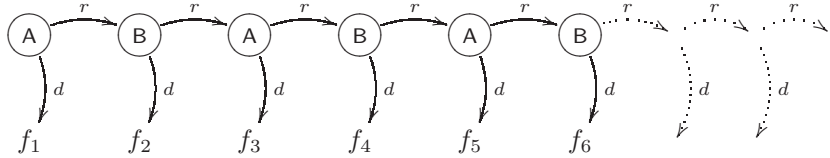
\begin{displaymath}
  \xymatrix{
*++[o][F]{\Al} \ar@/^/[r]^r \ar@/^/[d]^{d} &*++[o][F]{\Be} \ar@/^/[r]^r \ar@/^/[d]^{d}
&*++[o][F]{\Al} \ar@/^/[r]^r \ar@/^/[d]^{d} &*++[o][F]{\Be} \ar@/^/[r]^r \ar@/^/[d]^{d} 
&*++[o][F]{\Al} \ar@/^/[r]^r \ar@/^/[d]^{d} &*++[o][F]{\Be} \ar@{.>}@/^/[r]^r \ar@/^/[d]^{d} 
&\ar@{.>}@/^/[r]^r \ar@{.>}@/^/[d]^{d} &\ar@{.>}@/^/[r]^r \ar@{.>}@/^/[d]^{d} &\\
f_1&f_2&f_3&f_4&f_5&f_6&&
}
\end{displaymath}
At each step the agents have only two choices, namely to stop or to continue. Let us
call such a game, an \emph{infinipede}.

We introduce infinite games by means of equations.  Let us see how this applies to
define the $0,1$-game. First consider two payoff functions:
\begin{eqnarray*}
  f_{0,1} &=& \All "|->" 0, \Be "|->" 1\\
  f_{1,0} &=& \All "|->" 1, \Be "|->" 0
\end{eqnarray*}
we define two games 
\begin{eqnarray*}
  \GOI &=& \game{\All, \game{f_{0,1}}, \GIO}\\
  \GIO &=& \game{\Bel, \game{f_{1,0}}, \GOI}
\end{eqnarray*}
This means that we define an infinite sequential game $\GOI$ in which agent $\All$ is
the first player and which has two subgames: the trivial game $\game{f_{1,0}}$ and
the game $\GIO$ defined in the other equation.  The game $\GOI$ can be pictured as
follows:
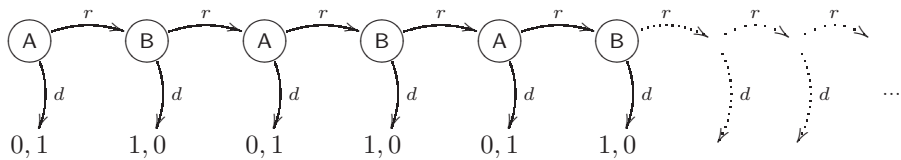
\begin{displaymath}
  \xymatrix{
*++[o][F]{\Al} \ar@/^/[r]^r \ar@/^/[d]^d &*++[o][F]{\Be} \ar@/^/[r]^r \ar@/^/[d]^d
&*++[o][F]{\Al} \ar@/^/[r]^r \ar@/^/[d]^d &*++[o][F]{\Be} \ar@/^/[r]^r \ar@/^/[d]^d 
&*++[o][F]{\Al} \ar@/^/[r]^r \ar@/^/[d]^d &*++[o][F]{\Be} \ar@{.>}@/^/[r]^r \ar@/^/[d]^d 
&\ar@{.>}@/^/[r]^r \ar@{.>}@/^/[d]^d &\ar@{.>}@/^/[r]^r\ar@{.>}@/^/[d]^d&\ar@{}[d]^{\ldots}\\
0,1&1,0&0,1&1,0&0,1&1,0&&&
}
\end{displaymath}
or more simply in Figure~\ref{fig:boucle}.a.
\begin{figure}[htb!]
  \centering
  \doublebox{\parbox{\iffullpage .8\textwidth\else \textwidth\fi}{ \begin{center} \begin{math}
      \begin{array}{ccc}
        \xymatrix@C=10pt{
          &\ar@{.>}[r]& *++[o][F]{\Al} \ar@/^1pc/[rr]^r \ar@/^/[d]^d 
          &&*++[o][F]{\Be} \ar@/^1pc/[ll]^r \ar@/^/[d]^d \\
          &&0,1&&1,0
        }
        &
        \xymatrix@C=10pt{
          &\ar@{.>}[r]& *++[o][F]{\Al} \ar@[blue]@{=>}@/^1pc/[rr]^r \ar@/^/[d]^d 
          &&*++[o][F]{\Be} \ar@/^1pc/[ll]^r \ar@[blue]@{=>}@/^/[d]^d \\
          &&1,0&&0,1
        } 
        &
        \xymatrix@C=10pt{
          &\ar@{.>}[r]& *++[o][F]{\Al} \ar@/^1pc/[rr]^r \ar@[blue]@{=>}@/^/[d]^d 
          &&*++[o][F]{\Be} \ar@[blue]@{=>}@/^1pc/[ll]^r \ar@/^/[d]^d \\
          &&0,1&&1,0
        }\\\\
        a\ (\GOI)  & b \ (\sioa) & c \ (\siob)
      \end{array}
\end{math}
\end{center}
}
}
\caption{The $0,1$-game and two equilibria seen compactly}
\label{fig:boucle}
\end{figure}
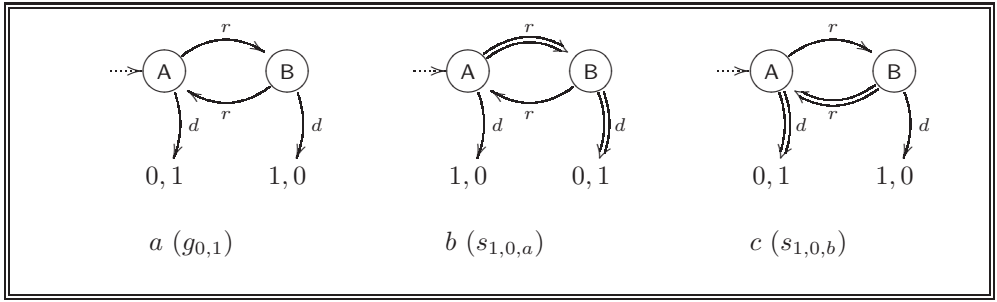

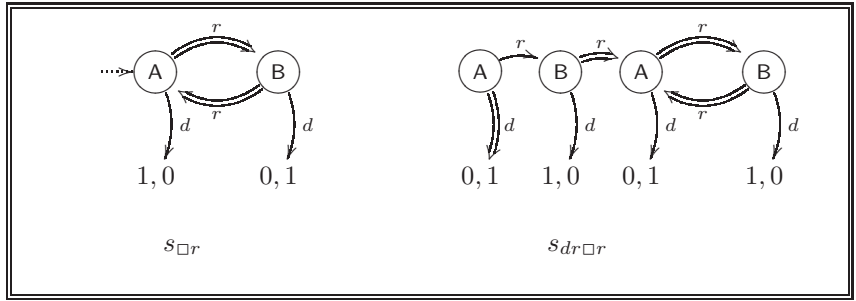
\begin{figure}[htb!]
  \centering
  \doublebox{\parbox{\iffullpage .65\textwidth\else .85\textwidth\fi}{ \begin{center} \begin{math}
     \begin{array}{ccc}
         \xymatrix@C=10pt{
          &\ar@{.>}[r]& *++[o][F]{\Al} \ar@[blue]@{=>}@/^1pc/[rr]^r \ar@/^/[d]^d 
          &&*++[o][F]{\Be} \ar@[blue]@{=>}@/^1pc/[ll]^r \ar@/^/[d]^d \\
          &&1,0&&0,1
        } 
        &\qquad\qquad
        \xymatrix@C=10pt{
          *++[o][F]{\Al} \ar@/^/[r]^r \ar@[blue]@{=>}@/^/[d]^d 
          &*++[o][F]{\Be} \ar@[blue]@{=>}@/^/[r]^r \ar@/^/[d]^d
          & *++[o][F]{\Al} \ar@[blue]@{=>}@/^1pc/[rr]^r \ar@/^/[d]^d 
          &&*++[o][F]{\Be} \ar@[blue]@{=>}@/^1pc/[ll]^r \ar@/^/[d]^d \\
          0,1&1,0&0,1&&1,0
        }\\\\
        s_{\Box r}  & s_{d r \Box r}
      \end{array}
\end{math}
\end{center}
}
}
\caption{Two examples of strategy profiles for the game $\GOI$}
\label{fig:alwr}
\end{figure}
 
From now on, we assume that we consider only strategy profiles $s$ whose game is the
0,1-game, that is $\mathsf{game}(s)=\GOI$. They are characterized by the following
predicates
\begin{eqnarray*}
  \SOI(s) &"<=(c)>"& s = \prof{\All, c, \prof{f_{0,1}}, s'} \wedge \SIO(s')\\
  \SIO(s) &"<=(c)>"& s = \prof{\Bel, c, \prof{f_{1,0}}, s'} \wedge \SOI(s').
\end{eqnarray*}
\begin{example}\label{ex:alwr}
  Here is a strategy profile (Figure~\ref{fig:alwr})
  \begin{displaymath}
    s_{\Box r} = \prof{\All,r, \prof{f_{0,1}}, \prof{\Bel,r,\prof{f_{1,0}},s_{\Box r}}  }
  \end{displaymath}
  where both agents continue forever. Notice that $ \SOI(s_{\Box r})$, $\neg (\conv{s_{\Box r}})$ and a
  fortiori $\neg (\Conv{s_{\Box r}})$. Said in words,
  \begin{enumerate}
  \item $s_{\Box r}$ has game $\GOI$,
  \item $s_{\Box r}$ is not convergent,
  \item $s_{\Box r}$ is not strongly convergent.
  \end{enumerate}

\end{example}
\begin{example}\label{ex:ddBoxr}
  Consider now the strategy profile (Figure~\ref{fig:alwr})
  \begin{displaymath}
    s_{d\Box r} = \prof{\All,d, \prof{f_{0,1}}, \prof{\Bel,r,\prof{f_{1,0}},s_{\Box r}}  }.
  \end{displaymath}
  \iffullpage\else\begin{sloppypar}\fi
    This time $ \SOI(s_{d\Box r})$, $\conv{s_{d\Box r}}$ and $\neg (\Conv{s_{dd\Box
        r}})$. Said in words,
    \begin{enumerate}
    \item $s_{d\Box r}$  has game $\GOI$,
    \item $s_{d\Box r}$  is convergent,
    \item $s_{d\Box r}$ is not strongly convergent.
    \end{enumerate}

We have $\widehat{s_{d\Box r}}(\Al) = 0$ and ${\widehat{s_{dd\Box
          r}}(\Be) = 1}$. But $s_{d\Box r}$ is not strongly convergent since
    $\prof{\Bel,r,\prof{f_{1,0}},s_{\Box r}}$ is not convergent.
   \iffullpage\else\end{sloppypar}\fi
\end{example}

Notice that the $0,1$-game we consider is somewhat a zero-sum game, but we are not
interested in this aspect.  Moreover, a very specific instance of a $0,1$ game has
been considered (by Ummels~\cite{DBLP:conf/fossacs/Ummels08} for instance), but these
authors are not interested in the general structure of the game, but in a specific
model on a finite graph, which is not general enough for our taste.  Therefore for
Ummels the $0,1$-game
is not a direct generalization of finite sequential games (replacing induction by
coinduction) and not a framework to study escalation.

\section{Subgame perfect equilibria}
\label{sec:subg-pertf-equil}

Among the strategy profiles, we can select specific ones that are called
\emph{subgame perfect equilibria}~\cite{selten65:_spiel_behan_eines_oligop_mit}.
Subgame perfect equilibria are specific strategy profiles that fulfill a predicate
\textsf{SPE}.  This predicate relies on another predicate \textsf{PE} which checks a
local property.
\[
\begin{array}{ll}
\mathsf{PE}(s) \quad "<=>" \quad \Conv{s}\  &\wedge\ s = \prof{p, d, s_d, s_r} "=>"
\widehat{s_d}(p) \ge \widehat{s_r}(p)\\
& \wedge\ s = \prof{p, r, s_d, s_r} "=>"
\widehat{s_r}(p) \ge \widehat{s_d}(p)
\end{array}
\] %
A strategy profile is a subgame perfect equilibrium if the property \textsf{PE}
holds always:
\[\mathsf{SPE} = \Box \mathsf{PE}.\]
\begin{example}\label{ex:spe}
  In Example~\ref{ex:stratprof} we have $\mathsf{SPE}(s_1)$, $\mathsf{SPE}(s_2)$ and
  $\neg \mathsf{SPE}(s_3)$.  In Example~\ref{ex:alwr}, %
  $\neg \mathsf{SPE}(s_{\Box r})$ since $\neg \Conv{(s_{\Box r})}$.
\end{example}
We may now wonder what the subgame perfect equilibria of the 0,1-game are.  We
present four of them in Figure~\ref{fig:boucle}.b, Figure~\ref{fig:boucle}.c and Figure~\ref{fig:ABAB}.  But there
are others.  To present them, let us define a predicate \emph{``\All{} continues and
  \Bel{} eventually stops''}
\begin{eqnarray*}
  \AcBes (s)  "<=(i)>"   s=\prof{p,c,\prof{f}, s'} &"=>"&
  (p =  \All \wedge f = {f_{0,1}} \wedge c=r \wedge \AcBes(s')) \vee\\ 
  &&(p = \Bel \wedge f={f_{1,0}} \wedge (c=d \vee \AcBes(s'))
\end{eqnarray*}
\begin{proposition}\label{prop:AcBes}
$(\SIO(s) \vee \SOI(s))  "=>" \AcBes(s) "=>" \widehat{s} = f_{1,0}$
\end{proposition}
\begin{proof}
Assume $\SIO(s) \vee \SOI(s)$.
  If $s=\prof{p,c,\prof{f}, s'}$, then $\SOI(s') \vee \SIO(s')$.  Therefore if
  $\AcBes(s')$, by induction, $\widehat{s'} = f_{1,0}$. By cases:
 \begin{itemize}
 \item If $p=\All \wedge c=r$, then $\AcBes(s')$ and by definition of $\widehat{s}$, we have $\widehat{s} =
   \widehat{s'}=f_{0,1}$
 \item if $p=\Bel \wedge c=d$, the $\widehat{s}=\widehat{\prof{f_{1,0}}} = f_{1,0}$.\pagebreak[2]
 \item if $p=\Bel \wedge c=r$,  then $\AcBes(s')$ and by definition of $\widehat{s}$,
   $\widehat{s}=\widehat{s'} = f_{1,0}$.
 \end{itemize}
\end{proof}
Like we generalize \textsf{PE}  to \textsf{SPE} by applying the modality $\Box$, we
generalize $\AcBes$ into $\SAcBes$ by stating:
\[\SAcBes = \Box\AcBes.\]
There are at least two profiles which satisfies $\SAcBes$ namely $\sioa$ and $\siob$
which have been studied in~\cite{DBLP:journals/corr/abs-1112-1185} and pictured in
Figure~\ref{fig:boucle}:
\begin{displaymath}
  \begin{array}{l@{\qquad\qquad}l}
    \begin{array}{lcl}
      \sioa &"<=(c)>"& \prof{\All, r, \prof{f_{0,1}}, \siob}\\
      \soia &"<=(c)>"& \prof{\All, d, \prof{f_{0,1}}, \soib}
    \end{array}
&
    \begin{array}{lcl}
      \siob &"<=(c)>"& \prof{\Bel, d, \prof{f_{1,0}}, \sioa}\\
      \soib &"<=(c)>"& \prof{\Bel, r, \prof{f_{1,0}}, \soia}
    \end{array}
  \end{array}
\end{displaymath}
In Figure~\ref{fig:ABAB}, we give other strategy profiles which fulfill the predicate
$\SAcBes$.  For the first one we draw only the beginning of the strategy profile, but
the reader can imagine that he continues a strategy profile in which $\All$ always
continues whereas $\Bel$ does not always continue, in other words, $\Bel$ stops infinitely
often.
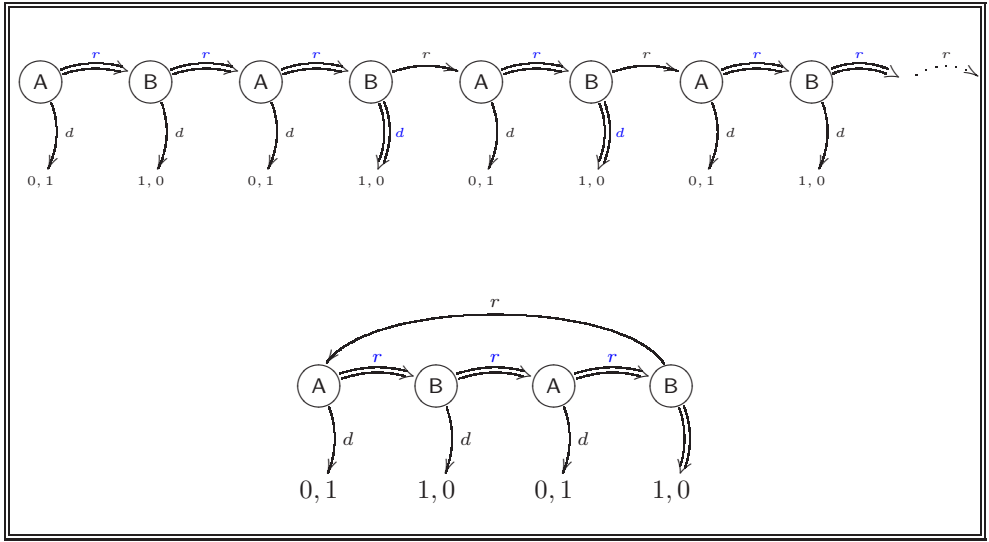
\begin{figure}[ht]
  \centering 
  \doublebox{\parbox{\textwidth}{ 
  \iffullpage 
   \begin{displaymath}
    \xymatrix{
      \nodA\flr{r}&\nodB\flr{r}&\nodA\flr{r}&\nodB\fld{r}&\nodA\flr{r}&\nodB\fld{r}%
      &\nodA\flr{r} &\nodB\flr{r}&\nodA\flr{r} &\ar@{.>}@/^/[r]^r&\\
      0,1&1,0&0,1&1,0&0,1&1,0&0,1&1,0& 0,1&&}
  \end{displaymath}
  \else
  \begin{tiny}
       \begin{displaymath}
    \xymatrix{
      \nodA\flr{r}&\nodB\flr{r}&\nodA\flr{r}&\nodB\fld{r}&\nodA\flr{r}&\nodB\fld{r}%
      &\nodA\flr{r} &\nodB\flr{r}&\ar@{.>}@/^/[r]^r&\\
      0,1&1,0&0,1&1,0&0,1&1,0&0,1&1,0&&}
  \end{displaymath}
  \end{tiny}
  \fi

\bigskip

\begin{displaymath}
    \xymatrix{
      \nodA\flr{r}&\nodB\flr{r}&\nodA\flr{r}&
      \nodB\ar@[blue]@/^/@2[d]\ar@/_2pc/@(u,u)[lll]_r\\
      0,1&1,0&0,1&1,0}
  \end{displaymath}
}}
  \caption{Other equilibria of the $0,1$ game.}
  \label{fig:ABAB}
\end{figure}
\begin{proposition}\label{prop:SAcBes-conv}
  \(\SAcBes(s) "=>" \Conv{s}.\)
\end{proposition}
We may state the following proposition.  
\begin{proposition}
  \(`A s, (\SOI(s) \vee \SIO(s)) "=>" (\SAcBes(s) "=>"  \mathsf{SPE}(s)).\)
\end{proposition}
\begin{proof}
Since \textsf{SPE} is a coinductively defined predicate, the proof is by coinduction.

  Given an $s$, we have to prove \(`A s, \Box \AcBes(s) \wedge (\SOI(s) \vee \SIO(s))
  "=>" \Box\mathsf{PE}(s).\)

  For that we assume $\Box \AcBes(s) \wedge (\SOI(s) \vee \SIO(s))$ and in addition
  (coinduction principle) \(\Box\mathsf{PE}(s')\) for all strict subprofiles $s'$ of
  $s$ and we prove $\mathsf{PE}(s)$. In other words, $\Conv{s} \wedge \prof{p, d, s_d,
    s_r} "=>" \widehat{s_d}(p) \ge \widehat{s_r}(p)\ \wedge\ \prof{p, r, s_d, s_r}
  "=>" \widehat{s_r}(p) \ge \widehat{s_d}(p). $

By Proposition~\ref{prop:SAcBes-conv}, we have $\Conv{s}$.

By Proposition~\ref{prop:AcBes}, we know that for every subprofile $s'$ of a profile
$s$ that satisfies $\SIO(s) \vee \SOI(s)$ we have $\widehat{s'} = f_{1,0}$ except
when $s'=\prof{f_{0,1}}$.  Let us prove $\prof{p, d, s_d, s_r} "=>" \widehat{s_d}(p)
\ge \widehat{s_r}(p)\ \wedge\ \prof{p, r, s_d, s_r} "=>" \widehat{s_r}(p) \ge
\widehat{s_d}(p). $ Let us proceed by case:
  \begin{itemize}
  \item $s=\prof{\All, r, \prof{f_{0,1}}, s'}$. Then $\SOI(s)$ and $\SIO(s')$. Since
    $\Box \AcBes(s)$, we have $\AcBes(s')$, therefore $\widehat{s'} = f_{1,0}$ hence
    $\widehat{s'}(\Al) = 1$ and $f_{0,1}(\Al) = 0$, henceforth $\widehat{s'}(\Al) \ge
    f_{0,1}(\Al).$
 \item $s=\prof{\Bel, r, \prof{f_{1,0}}, s'}$. Then $\SIO(s)$ and $\SOI(s')$. Since
    $\Box \AcBes(s)$, we have $\AcBes(s')$, therefore $\widehat{s'} = f_{1,0}$ hence
    $\widehat{s'}(\Be) = 0$ and $f_{1,0}(\Be) = 0$, henceforth $\widehat{s'}(\Be) \ge
    f_{1,0}(\Be).$
  \end{itemize}
\end{proof}
Symmetrically we can define a predicate $\BcAes$ for ``$\Bel$ continues and $\All$
eventually stops'' and a predicate $\SBcAes$ which is $\SBcAes = \Box\ \BcAes$ which
means that $\Bel$ always continues and $\All$ stops infinitely often. With the same
argument as for $\SAcBes$ we can conclude :
\begin{proposition}
  \(`A s, (\SOI(s) \vee \SIO(s)) "=>" \SBcAes(s) "=>" \mathsf{SPE}(s).\)%
\end{proposition}
\begin{lemma}
Assume $\SOI(s)$ or $\SIO(s)$, then  \(\mathsf{SPE}(s) "=>" (\SAcBes(s)\vee \SBcAes(s))\)
\end{lemma}
\begin{proof}
  By contradiction.  Assume $(\neg \SAcBes(s)) \wedge (\neg \SBcAes(s))$.  This means that one of the following
  statements are fulfilled.
  \begin{itemize}
  \item \emph{There exist $s_{\Be}$ and $s_{\Al}$ such that $s_{\Al} =
      \prof{\All,d,\prof{f_{0,1}},s_{\Al}'}\precsim s_{\Be} =
      \prof{\Bel,d,\prof{f_{1,0}},s_{\Be}'}\precsim s$ and there are only ``$r$'s''
      between $\Bel$ and $\All$.}  Notice that $\widehat{s_{\Be}'}(\Be) = 1$ while
    $\widehat{\prof{f_{1,0}}}(\Be)=0$.  Since $\mathsf{SPE}(s)$ then
    $\mathsf{SPE}(s_{\Be})$ therefore $s_{\Be} =
    \prof{\Bel,d,\prof{f_{1,0}},s_{\Be}'}$ implies $\widehat{s_{\Be}'}(\Be) \le
    \widehat{\prof{f_{1,0}}}(\Be)$ which is a contradiction.
  \item \emph{There exist $s_{\Al}$ and $s_{\Be}$ such that $s_{\Be} =
      \prof{\Bel,d,\prof{f_{1,0}},s_{\Al}'}\precsim s_{\Al} =
      \prof{\All,d,\prof{f_{0,1}},s_{\Be}'}\precsim s$ and there are only ``$r$'s''
      between $\All$ and $\Bel$.}  The contradiction is obtained like above. 
  \item \emph{$s_{\Box r} \precsim s$}, which means that eventually $\All$ and $\Bel$
    continue forever and which is in contradiction with
    $\mathsf{SPE(s)}$ since $\neg \mathsf{SPE}(s_{\Box r})$ (see Example~\ref{ex:spe}).
  \end{itemize}
\end{proof}
$\SAcBes \vee \SBcAes$ fully characterizes \textsf{SPE} of
0,1-games, in other words.
\begin{theorem}
  \(`A s, (\SOI(s) \vee \SIO(s)) "=>" (\SAcBes(s) \vee \SBcAes(s) "<=>"
  \mathsf{SPE}(s)).\)
\end{theorem}
\section{Nash equilibria}
\label{sec:Nash}
Before talking about escalation, let us see the connection between subgame perfect
equilibrium and Nash equilibrium in a sequential game.  In \cite{osborne04a}, the
definition of a Nash equilibrium is as follows: \textit{A Nash equilibrium is
  a``pattern[s] of behavior with the property that if every player knows every other
  player's behavior she has not reason to change her own behavior''} in other words,
\textit{``a Nash equilibrium [is] a strategy profile from which no player wishes to
  deviate, given the other player's strategies.'' }.  The concept of deviation of
agent $p$ is expressed by a binary relation we call
\emph{convertibility}\footnote{This should be called perhaps \emph{feasibility}
  following~\cite{rubinstein06:microec} and
  \cite{lescanne09:_feasib_desir_games_normal_form}} and we write \convp.  It is
defined inductively as follows:

    \[ \prooftree s\sbis s'
    \justifies s \convp s'
    \endprooftree
    \]

 \[
    \prooftree s_1 \convp s_1'\qquad s_2 \convp s_2' %
    \justifies \prof{p, c, s_1, s_2} \ \convp \ \prof{p, c', s_1', s_2'}
    \endprooftree
    \]

    \[
    \prooftree 
  s_1 \convp s_1'\qquad s_2 \convp s_2' %
  \justifies \prof{p', c, s_1, s_2 } \ \convp \ \prof{p', c, s_1', s_2'}
  \endprooftree
  \]

We define the predicate \textsf{Nash} as follows:
\[\mathsf{Nash}(s) "<=>" `A p,`A s', s \convp s' "=>" \widehat{s}(p) \ge
\widehat{s'}(p').\]

The concept of Nash equilibrium is more general than that of subgame perfect
equilibrium and we have the following result:
\begin{proposition}
  $\mathsf{SPE}(s) "=>" \mathsf{Nash}(s)$.
\end{proposition}
The result has been proven in COQ and we refer to the script
(see\cite{DBLP:journals/acta/LescanneP12}):

\centerline{\url{http://perso.ens-lyon.fr/pierre.lescanne/COQ/EscRatAI/}}

\centerline{\url{http://perso.ens-lyon.fr/pierre.lescanne/COQ/EscRatAI/SCRIPTS/}}

Notice that we defined the
convertibility inductively, but a coinductive definition is possible.  But this would
give a  more restrictive definition of Nash equilibrium.

\section{Escalation}
\label{sec:escalation}

Escalation in a game with a set $\textsf{P}$ of agents occurs when there is a tuple
of consistent strategies $(st_p)_{p`:\Player}$ such that its sum is not convergent,
in other words, $\neg~ \conv{\displaystyle (\bigoplus_{p`:\Player} st_p)}$. Said
differently, it is possible that the agents have all a private strategy which
combined with those of the others makes a strategy profile which is not convergent,
which means that the strategy profile goes to infinity when following the choices.
Notice the two uses of a strategy profile: first, as a subgame perfect equilibrium,
second as a combination of the strategies of the agents.

Consider the strategy:
\begin{eqnarray*}
  st_{\Al,\infty} &=& \strat{r, \strat{f_{0,1}}, st_{\Al,\infty}'}\\
  st_{\Al,\infty}' &=& \strat{\Be, \strat{f_{1,0}}, st_{\Al,\infty}}
\end{eqnarray*}
and its twin
\begin{eqnarray*}
  st_{\Be,\infty} &=& \strat{\Al, \strat{f_{0,1}}, st_{\Be,\infty}'}\\
  st_{\Be,\infty}' &=& \strat{r, \strat{f_{1,0}}, st_{\Be,\infty}}.
\end{eqnarray*}
Moreover, consider the strategy profile:
\begin{eqnarray*}
  s_{\Al, \infty} & = & \prof{\All, r, \prof{f_{0,1}}, s_{\Be, \infty}}\\
  s_{\Be, \infty} & = & \prof{\Bel, r, \prof{f_{1,0}}, s_{\Al, \infty}}.
\end{eqnarray*}
\begin{proposition}~\label{prop:escal}
  \begin{enumerate}
  \item $st_{\Al,\infty}\full{\textsf{\scriptsize \sf A}}$,
 \item $st_{\Be,\infty}\full{\textsf{\scriptsize \sf B}}$,
  \item $\mathsf{st2g}(st_{\Al,\infty}, \Al) = \mathsf{st2g}(st_{\Be,\infty}, \Be) = \GOI,$
  \item $\mathsf{game}(s_{\Al, \infty}) = \GOI$,
  \item $st_{\Al,\infty} \oplus st_{\Be,\infty} = s_{\Al, \infty}$,
  \item $\neg\ \conv{s_{\Al, \infty}}$.
  \end{enumerate}
\end{proposition}
\begin{proof}
  The first statements are proved by coinduction on the definition of
  $\full{\textsf{\scriptsize \sf A}}$, \textsf{st2g}, \textsf{game} and
  $st_{\Al,\infty}$.  The last statement is by induction on the definition of
  $\downarrow$.
\end{proof}

Proposition~\ref{prop:escal} can be said in words as follows:
\begin{enumerate}
\item $st_{\Al,\infty}$ is full for $\Al$,
\item $st_{\Be,\infty}$ is full for $\Be$,
\item  $st_{\Al,\infty}$ and  $st_{\Be,\infty}$ have game $\GOI$,
\item $s_{\Al, \infty}$ has game $\GOI$,
\item strategy  $st_{\Be,\infty}$ plus strategy $st_{\Be,\infty}$ yields profile
  $s_{\Al, \infty}$,
\item profile  $s_{\Al, \infty}$ is not convergent.
\end{enumerate}

$ st_{\Al,\infty}$ and $ st_{\Be,\infty}$ are both correct since they are built
using choices, namely~$r$, dictated by subgame perfect equilibria\footnote{Recall
  that our concept of intelligent choice is that of a subgame perfect equilibrium, as it
  generalizes backward induction, which is accepted following Aumann~\cite{aumann95}
  as the criterion of rationality for finite game.}  which start with $r$.  Another
feature of $0,1$-game is that no agent has a clue for what strategy the other agent
is using.  Indeed after $k$ steps, $\All$ does not know if $\Bel$ has used a strategy
derived of equilibria in $\SAcBes$ or in $\SBcAes$.  In other words, $\All$ does not
know if $\Bel$ will stop eventually or not and vice versa. The agents can draw no
conclusion of what they observe.  If each agent does not believe in the threat of the
other she is naturally led to escalation.

\section{Relativity: are agents really rational?}
\label{sec:really}

\rightline{\parbox{8cm}{\begin{it} Observers are less cognitively busy and more open
      to information than actors.
    \end{it}}}  %
  \medskip%
  \rightline{Daniel Kahneman~\cite{kahneman11:_think_fast_slow}}

\medskip

Rationality as observed by an \emph{outsider} is not the same as rationality seen by
an
\emph{insider}~\cite{kahneman11:_think_fast_slow,klapisch11:_my_piece_pie,chandor12:_margin_call}.
Since
coalgebras are the mathematical tool for observation~\cite{jacobs12:_introd_coalg} a
coalgebra approach does not come as a surprise.

In this paper we would like to use the following definition: ``An agent is rational\footnote{seen
  from inside} if she is motivated by maximizing her own payoff'', but since this leads to
debate, we prefer to say that an agent is intelligent in this case.   In infinite
extensive games, this translates in saying that an agent is intelligent if she adopts a
strategy consistent with a subgame perfect equilibrium as an extension of backward
induction.  This explains the behavior of traders especially when they enter
escalation.  But is this compatible with common sense?  More precisely whereas Howie
Hubler (who lost $8.67\times 10^{9}\$$ for Morgan Stanley), J\'er\^ome Kerviel (who
lost $6.95\times 10^{9}\$$ for Soci\'et\'e G\'en\'erale) or Brian Hunter (who lost
$6.69\times 10^{9}\$$ for Amaranth Advisors) were intelligent agents when they acted,
they are seen clearly as stupid by an external observer.  In other words, agents
reason intelligently in an irrational escalation, which means that an agent can be
rational in her closed world and seen irrational from outside.  Is
this consistent?  Yes and this allows us drawing three conclusions:
\begin{itemize}
\item Insider view differs from outsider observation.
\item According to
  K. E. Stanovich~\cite{stanovich2010intelligence,stanovich11:_ration}, there are
  \emph{two levels in Kahneman System~II} (the effortful and slow part of the
  mind)~\cite{kahneman11:_think_fast_slow}: \emph{algorithmic mind} which is the
  ability of agents to reason perfectly and logically in a deductive
  system\footnote{In infinite game theory, the deductive system includes coinduction
    and \textsf{SPE}'s.} (called \emph{mindware} by Keith Stanovich) and
  \emph{reflective mind} (or \emph{epistemic mind}) which is the ability of an agent
  to reconsider her believes.  For instance, reflective mind is visible when the
  agent realizes that the world (the game) is finite or when the agent realizes that
  maximizing her own profit is no more the main aim and that the survival of her
  company should be taken into account. In both cases, she changes her belief in an
  infinite world or in maximizing her own profit and adds or modifies one or more
  axiom(s) founding her deductive system, as part of her belief.
\item Consistently with the previous statement, an agent who is involved in an
  escalation can be considered as having a \emph{short term vision}, whereas the
  observer has a \emph{long term vision}.  The agent can also be considered as having
  a local vision, whereas the observer has a \emph{global vision}.
\end{itemize}

\section{Mindware revisited}

We propose to extend the concept of mindware proposed by Stanovich which itself
refines system II (slow thinking) of Kahneman~\cite{kahneman11:_think_fast_slow}.
For Stanovich, the deductive system of the mindware is always the same, namely
\emph{classical first order logic}, only its implementation in the mind of the agent
may vary and can be incomplete.  The reasoning of the mindware relies on believes that
can be changed by the reflective mind.  For us, the mindware may implement several
kinds of deductive systems and may change from one agent to the other. For instance,
it can implement classical first order logic, intuitionistic first order logic,
classical or intuitionistic first order logic extended with inductive reasoning or
extended with coinductive reasoning, higher order logic (intuitionistic or classical)
etc.  All those logics are equally acceptable, because they are consistent and from G\"odel theorem we know that no universal
system of deduction exists.  Like for Stanovich, believes may also be changed by the
reflective mind.  In his book \emph{What Intelligence Tests Miss: The Psychology of
  Rational Thought}~\cite{stanovich2010intelligence} Stanovich presents a few
examples to set his point. The fact that the mindware may implement several deductive
systems leads to reconsider the scope of his examples. This is the case for the
example that illustrates the beginning of chapter six:
\begin{em}
  \begin{quotation}
    Jack is looking at Anne but Anne is looking at George. Jack is married but George
    is not.  Is a married person looking at a unmarried person?

    A) Yes \quad B) No \quad C) Cannot be determined

    Answer A, B, or C. 

    [...]

    The vast majority of people answer C.
  \end{quotation}
\end{em}
In others words, the vast majority of people, me included, show on this example that
they have a correct mindware based on a constructive logic, for instance, on
intuitionistic logic, because such a logic is easier to use.  Therefore they
answer~C.  Stanovich claims that such an answer is incorrect and calls us
\emph{cognitive misers}.  Actually when reading
further~\cite{stanovich2010intelligence} I understood that I was supposed to use a
more sophisticated logic, especially that I should use the \emph{excluded middle}, in
other words, that I should use classical logic, because\emph{ ``most people can carry
  fully disjunctive reasoning when they are explicitly \emph{told} that it is
  necessary'' }(\cite{stanovich2010intelligence} p.~71).  Therefore I changed my
belief by adding \emph{p~or not p} and I answered~A.  But this requires more
calculation.  Indeed \emph{p or not~p} can be specialized into \emph{Anne is married
  or Anne is not married} from which we draw \emph{Anne is married and George is not
  married and Anne is looking at George or Jack is married and Anne is not married
  and Jack is looking a Anne}.  We can abstract this into \emph{there exists x and
  there exists y such that x is married and y is not married and x is looking at y or
  there exists x and there exists y such that x is married and y is not married and x
  is looking at y} which simplifies into \emph{there exists x and there exists y such
  that x is married and y is not married and x is looking at y} and leads to answer
A.  However if the above question is completed by
\begin{em}
  \begin{quotation}
    Choose  one of the followings

    A) Anne is married 

    B) Anne is not married 

    C) I don't know.
  \end{quotation}
\end{em}

A rational person including a person who answered A at the first question will answer
C at this question.  More precisely, answering A at the first question does not allow the agent to justify
her answer by exhibiting a pair of persons such that one is married and looking at
the other who is not married. Unable to justify their answer, the
majority of people choose C at the first question as well. 

To introduce disrationality, chapter two of~\cite{stanovich2010intelligence} starts
with the case of John Allen Paulos who was involved in a classical
escalation~\cite{paulos03:_mathem_plays_stock_market}.  Actually John Allen Paulos,
professor of mathematics at Temple University, is a typical person with a sound
algorithmic mind using coinduction knowingly or not\footnote{Actually Paulos uses
  coinduction unknowingly, rather invoking a kind of invariant}.  Actually he was
disrational because he used improperly his reflective mind and did not change his
belief in an ``everlasting'' Worldcom company (which eventually bankrupted in 2002)
and in its eventual restart.  Paulos is obviously ``foolish'', if we say that someone
is foolish, if he has a perfect algorithmic mind, but a faulty reflective mind, in
other words if he is intelligent, but not rational.

The above comments strengthen Stanovich's distinction between algorithmic mind and
reflective mind, but make the delineation of rationality harder and its evaluation
difficult, because attributing a ``rationality quotient'' requires first to determine
the deduction system implemented in the mindware and its strength, then to appreciate
the ability of the agent of changing her believes.

\section{Which deductive system?}
\label{sec:which}

\rightline{\parbox{8cm}{\begin{it} I say that it is not illogical to think that the world is infinite.
    \end{it}}}  %
  \medskip%
  \rightline{Jose Luis Borges, \emph{The library of Babel} in~\cite{borges41:_ficcion}}.

\medskip

We noticed that several deductive systems can be used.  We may wonder what features a
deductive system should have. Let us tell some of them. First the language should contain a modality to enable
agents to express \emph{belief} ($B_a$) or \emph{knowledge}~($K_a$).  Moreover, the
excluded middle ($p\vee \neg p$) is clearly not mandatory. However an agent $a$  should
be able to state that she believes in the excluded middle by a statement like
\begin{displaymath}
  B_a \mathbf{(}(`A p: \mathsf{Proposition})\ p\vee \neg p\mathbf{)}.
\end{displaymath}
This requires a quantification over propositions, which allows also expressing
sentences like: \emph{We know there are known unknowns}
(see~\cite{DBLP:journals/amai/Lescanne06} Section~4):
\begin{displaymath}
 \bigwedge_{a`:\mathsf{Agent}} K_a((`E p: \mathsf{Proposition})\ K_a (\neg K_a(p))).
\end{displaymath}
Besides quantifications over propositions, quantifications over functions and sets
are required to express belief in \emph{finiteness} (\cite{DBLP:journals/corr/abs-1112-1185} Section~6):
\begin{displaymath}
  B_a((`A A: \emph{Set})\,(`A f: A "->" A)\, (Surjective (f) "=>" Injectve(f)))
\end{displaymath}
where \textit{Surjective} is a predicate which asserts \emph{surjectivity} and \textit{Injective}
is a predicate which asserts \emph{injectivity}.



\section{Conclusion}
\label{sec:concl}

In this paper, we have shown how to use coinduction in economics, more precisely in
economic game theory where it has not been used yet, or perhaps in a hidden form,
which has to be unearthed.  We have shown also that rational agents can be seen as
irrational by observers since observation changes the point of view, in particular on
rationality.  When Wolfgang Leininger writes (see citation in front of
Section~\ref{sec:escal}) that the fact that rational agents should not engage in an
escalation seems obvious, he means, ``obvious'' for an observer, not for the agents,
or perhaps he should have said that rational agents should not engage in an
escalation, but that intelligent agents could.  If agents are only \emph{intelligent},
the efficiency of the markets should then be revisited at the light of
escalation. Therefore coinduction is a possible way for rethinking economics.


\appendix

\section{Finite 0,1 games and the ``cut and extrapolate'' method}
\label{sec:finite}

We spoke about the ``cut and extrapolate'' method, applied in particular to the
dollar auction.  Let us see how it would work on the 0,1-game.  Finite games, finite
strategy profiles and payoff functions of finite strategy profiles are the inductive
equivalent of infinite games, infinite strategy profiles and infinite payoff
functions which we presented.  Notice that payoff functions of finite strategy
profiles are always defined.  Despite we do not speak of the same types~\footnote{In
  the sense of type theory} of objects, we use the same notations, but this does not
lead to confusion.  Consider two infinite families of finite games, that could be
seen as approximations of the 0,1-game: \iffullpage
\begin{displaymath}
\begin{array}{c@{\quad~~}c}
\begin{array}{lcl}
  F_{0,1} &=& \game{\Al, \game{f_{0,1}}, \game{\Be, \game{f_{1,0}}, F_{0,1}}}  \cup \{\game{f_{0,1}}\}
\end{array}
&%
\begin{array}{lcl}
  K_{0,1} &=& \game{\Al, \game{f_{0,1}}, K_{0,1}'}\\
  K_{0,1}' &=& \game{\Be, \game{f_{1,0}}, K_{0,1}} \cup \{\game{f_{1,0}}\}
\end{array}
\end{array}
\end{displaymath}
\else
\begin{displaymath}
\begin{array}{lcl}
  F_{0,1} &=& \game{\Al, \game{f_{0,1}}, \game{\Be, \game{f_{1,0}}, F_{0,1}}}  \cup \{\game{f_{0,1}}\}
\end{array}
\end{displaymath}
\begin{displaymath}
\begin{array}{lcl}
  K_{0,1} &=& \game{\Al, \game{f_{0,1}}, K_{0,1}'}\\
  K_{0,1}' &=& \game{\Be, \game{f_{1,0}}, K_{0,1}} \cup \{\game{f_{1,0}}\}
\end{array}
\end{displaymath}

\fi In $F_{0,1}$ we cut after $\Bel$ and replace the tail by $\game{f_{0,1}}$.  In
$K_{0,1}$ we cut after $\All$ and replace the tail by $\game{f_{1,0}}$.
Recall~\cite{vestergaard06:IPL} the predicate \emph{backward induction} shortened in
\textsf{BI}, which is the finite and inductive version of~\textsf{PE}.
\begin{eqnarray*}
  \mathsf{BI}(\game{f})\\
  \mathsf{BI}(\game{p, c, s_d, s_r}) &=& \mathsf{BI}(s_l) \wedge \mathsf{BI}(s_r)
  \wedge \\
  &&\prof{p, d, s_d, s_r} "=>"
\widehat{s_d}(p) \ge \widehat{s_r}(p)\ \wedge\\ && \prof{p, r, s_d, s_r} "=>"
\widehat{s_r}(p) \ge \widehat{s_d}(p)
\end{eqnarray*}
\begin{example}
  In Example~\ref{ex:stratprof} we have $\mathsf{BI}(s_1)$, $\mathsf{BI}(s_2)$ and
  $\neg \mathsf{BI}(s_3)$.
\end{example}
We consider the two families of strategy profiles:

\begin{displaymath}
\begin{array}{c@{\qquad\quad}c}
  \begin{array}{rcl}
    \mathsf{SF}_{0,1}(s) &"<=(i)>"& (s = \prof{\All, d , \prof{f_{0,1}}, \prof{\Bel,
        r, \prof{f_{1,0}}, s'}} \wedge \mathsf{SF}_{0,1}(s'))  \quad \vee \\ 
    && (s = \prof{\All, r , \prof{f_{0,1}}, \prof{\Bel,
        r, \prof{f_{1,0}}, s'}} \wedge \mathsf{SF}_{0,1}(s'))  \quad \vee \\ 
    && s = \prof{f_{0,1}}
    \\\\
    \mathsf{SK}_{0,1}(s) &"<=(i)>"& s = \prof{\All, r, \prof{f_{0,1}}, s'} \wedge \mathsf{SK}_{0,1}'(s')\\
    \mathsf{SK}_{0,1}'(s) &"<=(i)>"& (s = \prof{\Bel, d , \prof{f_{1,0}}, s'} \vee s = \prof{\Bel, r , \prof{f_{1,0}}, s'}) \wedge
    \mathsf{SK}_{0,1}(s') \quad \vee \\ &&s = \prof{f_{1,0}}
  \end{array}
\end{array}
\end{displaymath}
In $\mathsf{SF}_{0,1}$, $\Bel$ continues and $\All$ does whatever she likes and in
$\mathsf{SK}_{0,1}$, $\All$ continues and $\Bel$ does whatever she likes.  The
following proposition characterizes the backward induction equilibria for games in
$F_{0,1}$ and $K_{0,1}$ respectively and is easily proved by induction:
\begin{proposition}~
  \begin{itemize}
  \item $\mathsf{game}(s) `:F_{0,1} \wedge \mathsf{SF}_{0,1}(s) "<=>" \mathsf{BI}(s)$,
  \item $\mathsf{game}(s) `: K_{0,1}  \wedge \mathsf{SK}_{0,1}(s) "<=>" \mathsf{BI}(s)$.
  \end{itemize}
\end{proposition}
This shows that cutting at an  even or an odd position does not give the same strategy
profile by extrapolation.  Consequently the ``cut and extrapolate'' method does not
anticipate all the subgame perfect equilibria.  Let us add that when cutting we
decide which leaf to insert, namely $\game{f_{0,1}}$ or $\game{f_{1,0}}$, but we
could do another way obtaining different results.

\paragraph{0,1 game and limited payroll.}

To avoid escalation in the dollar auction, people require a \emph{limited payroll},
i.e., a bound on the amount of money handled by the agents, but this is inconsistent
with the fact that the game is infinite.  Said otherwise, to avoid escalation, they
forbid escalation.  We can notice that, in the 0,1-game, a limited payroll would not
prevent escalation, since the payoffs are anyway limited by~$1$.  In the same vein,
Demange~\cite{demange92:_ration_escal} adds, to justify escalation, a new
feature called \emph{joker} which is not necessary as we have shown.
\end{document}